\documentclass{sig-alternate}

\conferenceinfo{PODC'11,} {June 6--8, 2011, San Jose, California, USA.} 
\CopyrightYear{2011} 
\crdata{978-1-4503-0719-2/11/06} 
\title{The Impact of Memory Models on Software Reliability \\ in Multiprocessors
\titlenote{This paper is a full version of an extended abstract that appears in PODC 2011.}
}
\newfont{\myeaddfnt}{phvr8t at 10pt}

\numberofauthors{5} 
\author{
\alignauthor Alexander Jaffe \\
       \affaddr{University of Washington} \\
       \email{\myeaddfnt{ajaffe@cs.washington.edu}}
\alignauthor Thomas Moscibroda \\
       \affaddr{Microsoft Research} \\
       \email{\myeaddfnt{moscitho@microsoft.com}}
\and
\alignauthor Laura Effinger-Dean \\
       \affaddr{University of Washington} \\
       \email{\myeaddfnt{effinger@cs.washington.edu}}
\alignauthor Luis Ceze \\
       \affaddr{University of Washington} \\
       \email{\myeaddfnt{luisceze@cs.washington.edu}}
\alignauthor Karin Strauss \\
       \affaddr{Microsoft Research} \\
       \email{\myeaddfnt{kstrauss@microsoft.com}}
}

\newcommand{\squishlist}{
   \begin{list}{$\bullet$}
    { \setlength{\itemsep}{0pt}      \setlength{\parsep}{3pt}
      \setlength{\topsep}{3pt}       \setlength{\partopsep}{0pt}
      \setlength{\leftmargin}{1.5em} \setlength{\labelwidth}{1em}
      \setlength{\labelsep}{0.5em} } }
      
\newcommand{\squishend}{
    \end{list}  }

\newdef{defn}{Definition}
\newtheorem{thm}{Theorem}[section]
\newtheorem{lemma}[thm]{Lemma}
\newtheorem{claim}[thm]{Claim}
\newtheorem{corollary}[thm]{Corollary}

\newcommand{\ceil}[1]{\lceil #1 \rceil }
\newcommand{\floor}[1]{\lfloor #1 \rfloor }
\newcommand{\E}{\mathbb{E}}
\renewcommand{\P}{\mathbb{P}}

\renewcommand{\epsilon}{\varepsilon}

\newcommand{\lval}{\gamma}
\newcommand{\lvalvec}{\bar{\lval}}
\newcommand{\lvar}{\Gamma}
\newcommand{\lvarvec}{\bar{\lvar}}

\newcommand{\rgam}[1]{\lvalvec^{#1}}
\newcommand{\rsig}[1]{\sigma^{#1}}

\usepackage{ifthen}
\newboolean{fullversion}
\setboolean{fullversion}{true}

\usepackage{enumerate}

\newcommand{\msf}{\mathsf}

\newcommand{\ld}{\msf{LD}}
\newcommand{\st}{\msf{ST}}

\newcommand{\prgm}{S}

\newcommand{\permi}[2]{\pi_{#1}(#2)}

\newcommand{\permiinv}[2]{\pi_{#1}^{-1}(#2)}

\newcommand{\plen}{m}

\allowdisplaybreaks[1]

\begin{document}

\maketitle

\begin{abstract}
The memory consistency model is a fundamental system property characterizing a multiprocessor. The relative merits of strict versus relaxed memory models have been widely debated in terms of their impact on performance, hardware complexity and programmability. This paper adds a new
dimension to this discussion: the impact of memory models on
software reliability. By allowing some instructions to reorder, weak
memory models may expand the window between critical memory operations.
This can increase the chance of an undesirable thread-interleaving,
thus allowing an otherwise-unlikely concurrency bug to manifest. To
explore this phenomenon, we define and study a probabilistic model
of shared-memory parallel programs that takes into account such
reordering. We use this model to formally derive bounds on the \emph{vulnerability} to concurrency bugs of different memory models. Our results show that for $2$ concurrent threads, weaker memory models do indeed have a higher likelihood of allowing bugs. On the other hand, we show that as the number of parallel, buggy threads increases, the gap between the different memory models becomes proportionally insignificant, and thus the importance of using a strict memory model diminishes.
%
%
\end{abstract}





\section{Introduction}



A critically important property of a shared-memory multiprocessor is
its {\em memory consistency model}. There has been an enormous amount
of work on this subject, both in industry and academia. The
memory consistency model describes which values may be returned by a
load operation in a parallel or multi-threaded program. The strongest and most intuitive
model is {\em Sequential Consistency} (SC)~\cite{lamport79}. SC
imposes two requirements on the execution of parallel programs: first,
all processors must see the same {\em global order} of memory
operations, and second, the operations for a particular processor must
appear to execute in {\em program order}. This model is attractive for
its high level of programmability, but the strict constraints on
memory operation reordering rule out important optimizations such as
access buffering, pipelining, or dynamic scheduling, which improve
performance by hiding the latency of memory accesses. In order to
enable these aggressive optimizations, a wide variety of {\em relaxed
  memory models} have been proposed. Relaxed memory models allow the
reordering of certain types of memory operations at the cost of
increased programming complexity, since programmers need to explicitly
encode reordering restrictions to ensure correctness.




Historically, the vast literature on memory consistency models has
discussed a three-way trade-off between performance, hardware
complexity, and programmability. In this paper, we bring a new axis to
this discussion: {\em software reliability}. Software is inherently
unreliable, and is arguably becoming less reliable with pervasive
concurrency.  Concurrency bugs such as data races and deadlocks are
extremely common in practice, and can cause unexpected failures in
even production-level code.

In this paper, we investigate to what extent relaxed memory consistency
models further contribute to the unreliability of parallel software by
increasing the likelihood that concurrency bugs will manifest during
an execution. For this purpose, we study a new probabilistic model for the instruction reordering introduced by relaxed memory models, and analyze a canonical buggy program (specifically, an atomicity violation~\cite{qadeer03,     artho03, lu08}) with respect to this model. We compare three important memory consistency models: Sequential Consistency, Weak Ordering, and Total Store Order. We derive two interesting results for our model:
\squishlist
\item We show that for 2 (or any small constant number of) parallel threads, the bug is indeed more likely to manifest under weaker memory models. This is intuitive and follows from the following high-level argument: A typical concurrency bug, such as a data race, can manifest only during a short window of time. The reordering of operations caused by relaxed memory models may increase the size of this critical window, thus making the bug more likely to manifest. In the paper, we give precise bounds on this \emph{vulnerability} of the three memory models.
\item On the other hand, we show that as the number of parallel, buggy threads increases, the gap between the different memory models shrinks in proportion to the risk for even the strongest memory model. This implies that \emph{as the number of parallel threads in the system increases, the importance of using a strict memory model diminishes} (with regard to the software reliability metric we study in this paper).
\squishend

Notice that the latter result could have far-reaching implications on the choice of memory consistency models in future multi-core and massively parallel systems. Intuitively, one might expect that with more and more concurrent threads, stronger memory consistency models should be used in order to counter the generally increased likelihood of bugs. However, our results indicate that the opposite is the case: As the number of threads increase, the relative importance of having stronger memory models reduces to a minimum. The underlying reason is that the larger number of threads causes the likelihood that bugs occur to increase much more quickly than what even the strictest memory model is able to contain. That is, the asymptotic growth fundamentally works against using strict memory models as we increase the number of threads.

The technical content of our paper proceeds as follows. In Section~\ref{sec:model}, we introduce two
distinct random processes, each of which is a natural object of inquiry in
isolation. By combining them---treating the output of the first process as
the input to the second---we model the end-to-end behavior of program
execution. This allows us to answer our central question: how does the
probability that a canonical data race manifests vary across memory
models and quantity of threads?

The first process models the generation of a random program,
and the subsequent randomized reordering of instructions. Specifically, in Section~\ref{sec:critwindow}, we derive the
probability that a certain essential window of vulnerability between two
instructions widens. The second process enacts a random series of shifts on
a set of heterogenous segments of the integer line. We use the positions of these line segments
to model the interleaving of the vulnerable windows of the threads. In Section~\ref{sec:shift}, we
estimate the probability that each of these segments is shifted to mutually
disjoint positions. Finally, the two processes are combined together in Section~\ref{sec:together}
to derive overall bounds on the probability of bug manifestation, first for two threads, then for
a large number of threads. Due to lack of space, several proofs are omitted and deferred to the 
\ifthenelse{\boolean{fullversion}}
{appendix.}
{full version of this paper.}

\pagebreak

\section{Background \& Related Work}
\label{sec:background}

\vspace{-1.5mm}

\subsection{Memory Consistency Models}

Memory models are a key aspect of the hardware/software interface in
shared-memory multicore/mul\-ti\-pro\-ces\-sor systems. They determine
what values read memory operations are allowed to return by dictating
how memory operations are allowed to be reordered, as well as when
writes become visible to other processors. They have major
implications on the performance, design complexity and programmability
of multiprocessor systems and the programs that run on them.  Common
misunderstandings about memory models often lead to bugs that are very
difficult to find and fix, and can also lead to major performance
issues.  There exists a vast and rich line of literature on memory
models (a good tutorial overview is presented in~\cite{adve96}). Most
of the past work has focused on new memory
models~\cite{gharachorloo90, adve90, goodman89}, hardware
implementations~\cite{gharachorloo91a, gniady99, ceze07}, memory
models for popular languages such as Java~\cite{manson05} and
C++~\cite{cppmm}, and compiler optimizations~\cite{pensieve} and their relative merits~\cite{adve96, arvind06}.



\vspace{1mm}\noindent\textbf{Relaxed memory models:} The strongest memory model is
Lamport's \emph{Sequential Consistency} (SC)~\cite{lamport79}. In order to
enable important performance optimizations, a number of relaxed memory
models have been proposed in the literature, with varying degrees of
guarantees. One of the strongest examples is known as \emph{Total
  Store Order} (TSO)~\cite{sparcmanual}. In TSO, loads may execute
before stores that precede them in program order, as long as no data
dependency is violated. All other pairs of instructions must maintain
strict program order. This model encapsulates the natural case in
which stores are observed by remote processors in program order. Some
stores may take extra time to be observed after their execution, but
the local program is allowed to proceed. A similar, but slightly
weaker consistency model is \emph{Partial Store Order}
(PSO)~\cite{sparcmanual}, which also allows the reordering of stores
with respect to each other as long as they access distinct memory
locations. A significantly weaker consistency model is \emph{Weak
  Ordering} (WO)~\cite{dubois86, adve90}. The opposite extreme from
Sequential Consistency, WO allows any memory operations to reorder
with one another, as long as no data dependencies are violated. This
model allows for an equal amount of optimization as a uniprocessor,
but is also the most vulnerable to programmer error, since it requires
explicit {\em fences} to prevent unwanted reorderings.  Modern
processors typically support relaxed models. For example, the x86
memory model~\cite{X86AMD, X86intel} supports a model similar to TSO
and the IBM POWER architecture supports a form of WO.

The above memory consistency models follow a pattern: they can be defined by a subset
of the four ordered memory operation pairs, specifying which pairs are
allowed to reorder: For example, in the WO model, any two memory operations are allowed to be reordered;
in SC, no two memory operations are allowed to be reordered; and in the TSO model, no two memory
operations are allowed to be reordered, except that loads can reorder
before stores (see Table~\ref{table:basicmodels}).

\begin{table}[t!]
\begin{scriptsize}
\begin{center}
\begin{tabular}{|c|c|c|c||c|}
  \hline
  $\st/\st$ & $\st/\ld$ & $\ld/\st$ & $\ld/\ld$  & Name \\
  \hline
  & & &  & Sequential Consistency \\
  & X & &  & Total Store Order \\
  X & X & &  & Partial Store Order \\
  X & X & X & X  & Weak Ordering \\
  \hline
\end{tabular}
\caption{
  Important memory models. A ``X'' in column $\st/\ld$ means that the
  ordering restriction from stores to later loads can be relaxed,
  i.e., loads can complete before stores that precede them in program
  order.  With regard to our model in
  Section~\ref{sec:reordering_model}, this means that a $\ld$ can
  settle past (swap with) a preceding $\st$. Other columns are
  analogous. \vspace{-6mm}} \label{table:basicmodels}
\end{center}
\end{scriptsize}
\end{table}


Note that since in this paper we analyze a concurrency bug
involving multiple threads, we ignore store atomicity~\cite{arvind06},
which is tangential to our present analysis.
Moreover, we do not currently handle {\em fence} operations explicitly,\footnote{However, our shift process
in Section~\ref{sec:shift} can be used to simulate a behavior similar to that arising from the use of fences. } which
are used to restrict reorderings and are typically used for
synchronization. For that reason, we do not consider models such as Release Consistency
(RC)~\cite{gharachorloo90}, which differs mainly in the types of
fences supported. As we discuss in Section~\ref{sec:discussion}, it will be interesting
to extend our process to distinguish such memory models.

\subsection{Race Conditions}
\label{sec:background_races}

A common type of bug in shared-memory multithreaded programming is a
\emph{race condition}, which occurs when correctness depends on an
assumption about the order in which instructions from two or more
threads interleave. In particular, an \emph{atomicity
  violation}~\cite{qadeer03} occurs when the programmer assumes that
multiple instructions will execute as an atomic unit, but fails to
insert the proper synchronization. A recent study showed
that atomicity violations are extremely common in ``real world''
programs~\cite{lu08}.
Race conditions are often difficult to identify due to nondeterminism:
the program may behave correctly most runs, but fails only for
specific thread interleavings.


A canonical example of an atomicity violation is as follows:
\begin{center}
  \begin{small}
    \sf
      \begin{tabular}{l|l}
        Thread 1 &  Thread 2 \\ \hline
        1: int loc = x; &  1: int loc = x; \\
        2: loc = loc + 1; & 2: loc = loc + 1; \\
        3: x = loc; & 3: x = loc; \\
      \end{tabular}
  \end{small}
\end{center}
Here \textsf{x} is a shared variable (with \textsf{x = 0} initially)
and \textsf{loc} is local to each thread. Two threads simultaneously
try to increment \textsf{x} by loading its value into a local
variable, incrementing that local variable, then storing the updated
value back to \textsf{x}.  The programmer's intent is that $\mathsf{x
  = 2}$ after both threads finish executing. However, the program has
a race condition that can result in the spurious outcome $\mathsf{x =
  1}$.  For instance, suppose that the two threads interleave as
follows: (1) Thread~1 executes Lines~1 and 2; (2) Thread~2 executes
Lines~1 and 2; (3) Thread~1 executes Line~3; (4) Thread~2 executes
Line~3. This interleaving produces the final result $\mathsf{x = 1}$.
We say that the bug \emph{manifests} because the result did not match
programmer intent.

The standard solution for race conditions like the example above is to
protect the variable \textsf{x} with a lock. However, locking
protocols can be extremely complicated in large programs, and in practice,
a concurrency bug may easily slip past even the most experienced
programmers. Note that such bugs can manifest in any memory model, even
Sequential Consistency.


\section{Model}
\label{sec:model}


Our goal is to study how the use of different memory models
impacts the likelihood of an error occurring given a canonical atomicity
violation. In this section, we describe a model that allows us to formally analyze these likelihoods.
It is a probabilistic model of parallel
program executions under memory models that may permit reordering. At a
high level, we consider two or more threads which execute a simple program
containing an atomicity bug. The program consists of basic
memory operations (stores and loads). Depending on the memory model
under consideration, the
operations in each thread are then independently reordered via a random process
we call the \emph{settling process}. Finally, we use a thread interleaving model---the \emph{shift model}---
to model the execution of the program by
interleaving the instructions of different threads. The probability of the
bug manifesting is determined by analyzing how the operations from the
threads interleave. We show in this paper that, when
executing two threads, this probability crucially depends on the
underlying memory model. Yet, perhaps counter-intuitively, we show
that as the number of threads grows larger, the relative difference between the memory models becomes
smaller and smaller.

\begin{figure*}
  \includegraphics[width=\textwidth]{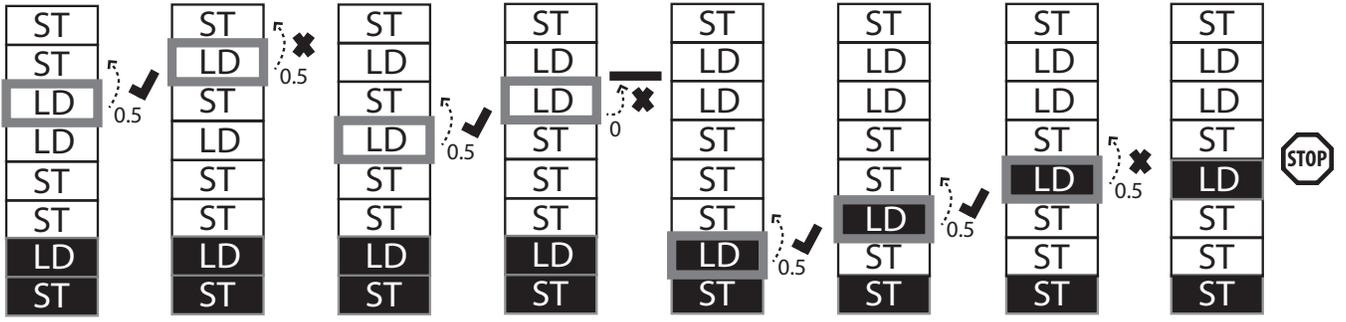}
  \caption{ An instantiation of the settling process under TSO. $\ld$s
    repeatedly settle upward with probability 1/2. If they fail to
    settle, or encounter another $\ld$, they stop permanently, and the
    next-lowest $\ld$ begins. The black boxes represent the critical
    instructions. The grey outlines indicate the currently settling
    instruction. The bottom four instructions in the final order form
    the \emph{critical window}.}
  \label{fig:settling}
\end{figure*}



\subsection{Program Model}
\label{sec:program_model}

We first describe a process for modeling a typical, randomly reordered
program. The process proceeds in two phases: program generation and
program reordering.

\subsubsection{Program Generation} \label{sec:generation}
We model an initial program based on the canonical atomicity violation bug
described in \S\ref{sec:background_races}. The program is a sequence
$S$ of memory operations $x_1$, $x_2$, $\ldots$, $x_{\plen}$,
$x_{\plen+1}$, $x_{\plen+2}$, where each $x_i$ has \emph{type} $\tau(x_i) \in
\{ \ld, \st \}$. $x_{\plen+1}$ and $x_{\plen+2}$ are Lines~1 and~3 of the
canonical bug, respectively. Since we are only concerned with memory
operations, we omit Line~2 (which accesses only the local variable
\textsf{loc}), and we will use the terms \emph{instruction} and
\emph{memory operation} synonymously in this
paper. We assume for simplicity that that only $x_{\plen+1}$ and $x_{\plen+2}$
access the same location.\footnote{If two instructions access the same
  location, they cannot reorder, so this assumption simplifies our
  analysis.} We will call $x_{\plen+1}$ the \emph{critical load} and
$x_{\plen+2}$ the \emph{critical store}. An \emph{initial program order}
$S_0$ starts with a random sequence of $\plen$ independently distributed
$\ld$ and $\st$ operations; $\tau(x_i)=\st$ with probability $p$ and
$\ld$ with probability $1-p$.
Furthermore, for convenience in the analysis, it will be useful to
approximate a very long program by letting $m \to \infty$.

\subsubsection{Instruction Reordering: The Settling Process}
\label{sec:reordering_model}

Different memory models allow for different forms of instruction
reorderings. We model this relaxation of program order using a
probabilistic \emph{settling process}.  This random process models
instruction reordering by taking a (random) initial program order as
input, and producing a reordering of that initial program. The
settling process takes into account which kinds of reorderings are
allowed by the memory consistency model under consideration, and
generates a random program order that is allowed to occur given the
kinds of reorderings.  In this section, we give an informal
description of the settling process; a formal definition is given in
\ifthenelse{\boolean{fullversion}}
{Appendix~\ref{app:modeldef}.}
{the full paper.}
Figure~\ref{fig:settling} presents a visualization of the settling
process.

Given an initial program order $\prgm_0$, the settling process
proceeds in $\plen+2$ rounds. In the $r$th round, (1) the program order
$\prgm_{r-1}$ from the end of the $(r-1)$st round is taken as the
input, and (2) the $r$th instruction is \emph{settled} in this program
order, which (3) creates the new program order $\prgm_{r}$. The final
output of the settling process is the program order $\prgm_{\plen+2}$
after settling the critical store $x_{\plen+2}$. Settling
the $r$th instruction in round $r$ of the process works as
follows. Instruction $x_r$ is recursively reordered (that is, swapped
in the current program order) with its preceding instruction
(initially, this is the instruction at position $r-1$), until a
reordering ``fails,'' in which case $x_r$ remains at its current
position in the program order. A reordering always fails if the memory
consistency model does not allow two operations of this type to be
reordered.  Otherwise, the reordering succeeds with some fixed
probability $s$, and fails with probability $1-s$.\footnote{A more
  general form of the settling model allows different nonzero
  probabilities for different kinds of reorderings, depending on the
  types of memory operations involved. For example $s_{\ld,\ld}$ can
  be different from $s_{\ld,\st}$, even if both are nonzero.} When a reordering fails, we
move onto the next round.

For ease of exposition, we will set both probabilities
$p$ (from program model) and $s$ to be $1/2$ in subsequent sections. However, note that
as long as $s$ and $p$ are constant, the key theorems and conclusions derived in this paper
remain fundamentally the same (though some of the numerical values change somewhat).

\vspace{1mm}

\textbf{Examples:} In SC, no instructions are allowed to be
reordered; hence $S_{\plen+2} = S_0$. In WO, all types of reorderings
are allowed, so, starting from instruction 2 in the initial program
order, each instruction is settled using a series of swaps with its
preceding instructions, until with probability $1-s$ a swap
fails. Then the next instruction is settled, and so forth. TSO relaxes
only the $\st \to \ld$ ordering, which in our model implies that a
$\ld$ may reorder with a preceding $\st$ with probability $s$, but all
other types of reorderings fail.

\vspace{1mm}

We will represent the result of a settling process by a permutation
on the indices. For thread $k$,
$\pi^{(k)}(i): [1, 2, \ldots, \plen+2] \rightarrow [1, 2, \ldots, \plen+2]$
maps the instruction starting at position $i$ to its final settled
position.

The settling process has two key features: (1) memory model
constraints are enforced (two operations can reorder only if allowed
by the memory model), and (2) reorderings that \emph{are} allowed
occur with a fixed likelihood. One effect of the latter property is
that in the final program order, most
instructions will not to move too far from their position in the
initial program order. The critical property of a memory consistency
model that we seek to capture is the \emph{degree to
  which individual instructions can reorder beyond other instructions,
  and thus move further away from their original position.}

\subsection{Thread Interleaving Model}
\label{sec:interleaving-model}


We describe a second high-level random process, which is used to
determine the interleaving of $n$ threads when they are executed
simultaneously on a multiprocessor. In fact, the process is quite
general, and may be of independent interest as a probabilistic
model. We first describe it in the abstract, then discuss how it will
be used to determine the effect of the program model's output on the
probability of bug manifestation.
\begin{defn} \label{def:shift} Consider a sequence of $n$ positive
  line segments originating at 0, having integer lengths $\bar{\gamma}
  = \gamma_1, \ldots, \gamma_n$. A \emph{shift process} translates the
  segments by i.i.d.\ geometric random variables $s_1, \ldots,
  s_n$. Then the random event of interest, called $A(\bar{\gamma})$,
  is the event that the segments are shifted such that all are
  mutually disjoint. That is,
  \[
  A(\bar{\gamma}) := [s_i, s_i + \gamma_i] \cap [s_j, s_j + \gamma_j] = \emptyset \ \forall \ i \neq j.
  \]
\end{defn}

In Section \ref{sec:shift}, we will analyze the probability of $A(\bar{\gamma})$ for
arbitrary segment lengths $\bar{\gamma}$. However, to connect this model to the task
at hand, we will go on to think of these segment lengths as the \emph{critical
windows} of reordered programs generated by the program model.

Recall that we study a canonical data race, for which correct execution
requires that each thread's pair of critical $\ld$ and critical $\st$ be
executed atomically. We thus refer to the sequence of instructions between
the critical $\ld$ and $\st$ (inclusively) as the \emph{critical window} of a
thread. We let $B_{\gamma}^{k}$ be the event that the final ordering of
thread $T_k$ inserts $\gamma$ instructions between the critical $\ld$ and $\st$,
(sometimes referred to as the \emph{critical window growth} of a memory model).
Manifestation of the data race corresponds exactly to the
event that when the reordered threads are executed in parallel, some pair of
critical windows are \emph{not} executed disjointly. We let $A$ refer to the event
that critical windows are disjoint. One can then think of $\Pr[B_{\gamma}^k]$
and $\Pr[A]$ as the two fundamental values we seek to characterize in this paper
- each a measure of the vulnerability of a memory model to this canonical data
race.

The shift model is used to simulate the parallel execution of the critical
windows of each thread, under the following assumptions.  All threads are
assumed to initially be identical copies of a single program, generated
randomly as in Section \ref{sec:generation}. Each thread is then independently
reordered according the process of Section \ref{sec:reordering_model}. We then
simulate the parallel execution of the reordered threads by placing the final
instruction of each critical window
the origin of the number line (here representing time in reverse, with 0 being the final
time step of execution), and using the shift model
of Definition \ref{def:shift} to model the \emph{varying
rates of execution} of each thread. After shifting, the execution of
each instruction is assumed to take one unit of time; instructions begin
and end synchronously across all threads, in lock-step. We assume that instructions
instantaneously read the current state of the
system at the beginning of the time step, and instantaneously commit
their changes at the end of the time step. In this way we ensure a clear
semantics for the state of the system at any given time: when a $\ld$
executes, it observes all the effects of any $\st$ that completed in a
time step preceding it.

We can now observe the circumstances in which a data
race manifests. There must be two threads such that,
subsequent to reordering, the final regions of time steps between the
critical $\ld$ and $\st$ (inclusive) overlap with one another.
In this case the data race must manifest, because one of the $\ld$s must
observe a value after (or simultaneous to) the other $\ld$ being observed,
but before the other $\st$ has committed.

A formal definition and a graphical visualization of the shift process
is in
\ifthenelse{\boolean{fullversion}}
{Appendix~\ref{sec:interleavingdef} (see Figure~\ref{fig:shift}).}
{the full paper.}

\section{The Critical Window}
\label{sec:critwindow}

In this section, we study what is perhaps the core component of our random
process, and the only one that directly distinguishes the memory models:
the reordering of instructions within an individual thread. In particular,
we are interested in the final distribution of the size of the
\emph{critical window} between the critical $\ld$ and $\st$. For the extreme memory
models of Sequential Consistency and Weak Ordering, we are easily able to
exactly characterize this distribution. The bulk of the technical challenge of
this section (and consequently of later sections) is in establishing results
for the more subtle model, Total Store Order. By carefully conditioning on
several auxiliary random variables, lower bounding complex algebraic terms
by their low-indexed values, and utilizing a bound on the
\emph{partition number} of certain integers, we derive rather sharp
approximations for the distribution of the critical window size.
These bounds will in subsequent sections be plugged into derived formulae for
the probability of bug manifestation, as a function of the thread
\emph{interleaving} process. Though the results in this section are tailored
specifically to the thread generation and reordering processes specified in
the previous section, it is worthwhile to observe how the asymptotics of the
overall bug manifestation probability will not depend delicately on the details
of this process.

We will be estimating the critical window growth, $\Pr[B^{k}_{\gamma}]$,
for a select set of memory models. Recall that $B^{k}_{\gamma}$ is the event that
the thread $T_k$ inserts $\gamma$ instructions between the critical $\ld$ and $\st$
in reordering. Because we will be considering a single fixed
thread in this subsection, we will refer to the event $B^{k}_{\gamma}$ by $B_{\gamma}$,
and the permutation $\pi^{(k)}$ by $\pi$. The first two memory models can be
considered a warmup, for the substantially more challenging case of Total Store
Order. All of these results are captured in the following theorem.

\begin{thm} \label{thm:crit} The critical window growth behaves
  according to the following functions:

\squishlist

\item \textbf{Sequential Consistency:} \label{thm:crit-seq}
\begin{equation*}
\Pr[B_{\gamma}] =
\begin{cases}
1 & \text{if } \gamma = 0,\\
0 & \text{if } \gamma > 0.
\end{cases}
\end{equation*}

\item \textbf{Weak Ordering:} \label{thm:crit-weak}
\begin{equation*}
\Pr[B_{\gamma}] =
\begin{cases}
2/3 & \text{if } \gamma = 0,\\
(2^{-\gamma})/3 & \text{if } \gamma > 0.
\end{cases}
\end{equation*}

\item \textbf{Total Store Order:} \label{thm:crit-total}
\begin{equation*}
\Pr[B_{\gamma}] =
\begin{cases}
2/3 & \text{if } \gamma = 0,\\
(6/7) \cdot 4^{-\gamma} + R(\gamma) \cdot 2^{-\gamma} & \text{if } \gamma > 0,
\end{cases}
\end{equation*}
for non-negative approximation term $R(\gamma) \leq \frac{2}{21}$.

\squishend

\end{thm}

Observe that the critical window grows at vastly different rates across
the models. Up to lower-order terms, the probability of a window size $\gamma$
is $2^{-\gamma}$ in Weak Ordering, $(2^{-\gamma})^2$ in Total Store Order,
and $0$ in Sequential Consistency. It remains to be seen in later sections
the extent to which this window size effects bug manifestation.

\begin{proof}[(Theorem \ref{thm:crit}---Sequential Consistency)]
~ \\
Under sequential consistency, no instruction is ever allowed to reorder.
Hence $\Pr[B_0] = 1$, and $\Pr[B_{\gamma}] = 0 \ \forall \gamma \neq 0$.
\end{proof}

We next consider the case of intermediate difficulty: Weak Ordering.

\begin{proof}[(Theorem \ref{thm:crit}---Weak Ordering)]
%
~ \\
Under weak ordering, all four ordered pairs of instruction types are allowed
to pass one another. Recall that we assume a strong normal form, in which all
possible swaps occur with probability 1/2. Hence in weak ordering, each
subsequent instruction continually moves up with probability 1/2, until it
ever fails to swap. This applies to the critical load and critical store as
well, with the exception that the critical store will never pass the critical
load, (because they access the same address). To calculate the probability,
we condition on the resting position of the critical $\ld$, which entails
a given resting position for the critical $\st$, for any $\gamma > 0$.
\begin{align*}
  \Pr[B_{\gamma}] & = \Pr[\pi(\plen+2) - \pi(\plen+1) = \gamma + 1] \\
  & = \sum_{i = \gamma}^{\infty} \Pr[\pi(\plen+1) = \plen+1-i] \\
  & \quad \quad \quad \cdot \Pr[\pi(\plen+2)=\plen +2 - i + \gamma | \\
    & \quad \quad \quad \quad \quad \quad \pi(\plen+1) = \plen+1-i] \\
  & = \sum_{i = \gamma}^{\infty} 2^{-(i+1)} 2^{-(i+1-\gamma)} 
 \;=\; \frac{2^{-\gamma}}{3}.
\end{align*}
We must handle the case of $\gamma = 0$ separately, because here the
critical $\st$ stops moving ``automatically,'' when it runs up against
the critical $\ld$.
\begin{align*}
\Pr[B_{0}] &= \sum_{i = 0}^{\infty} \Pr[\pi(\plen+1) = \plen+1-i] \\
& \quad \cdot \Pr[\pi(\plen+2)=\plen +2 - i | \pi(\plen+1) = \plen+1-i] \\
&= \sum_{i = 0}^{\infty} 2^{-(i+1)} 2^{-(i)}  \;=\; 2/3. \qed
\end{align*}
\end{proof}

Finally we turn to the far more challenging setting of Total Store Order.

\begin{proof}[(Theorem \ref{thm:crit}---Total Store Order)]
%
%
  ~ \\
  One of the strongest and most commonly used relaxed memory models,
  Total Store Order (TSO) only permits loads to swap with
  stores. Hence in calculating the distribution of window size, we
  need only consider the number of stores located directly before the
  critical load. Those stores will never move themselves, and the
  critical load can never swap past the first load above it. Moreover,
  the critical store never swaps with anything, so its final position
  is fixed.

However, deriving bounds on
$\Pr[B_{\gamma}]$ is difficult. \emph{$\ld$ operations may reorder
  past $\st$ operations, thus pushing longer sequences of $\st$
  operations together.} In this section we derive
bounds on the critical window growth for TSO, which is a core
technical contribution of this paper.
%
%
The proof is quite involved. Much difficulty arises
in gaining control over the relative positions of $\ld$s and $\st$s.
We outline the steps taken to estimate the critical window growth below.
The majority of these steps are non-trivial, and often involve a delicate
case analyses.

\paragraph{Proof Outline}
\squishlist
\item[1.] Express the critical window probability in terms of a series of new
random variables, $L_{\mu}$: the event that the second-to-last reordering
leaves exactly $\mu$ contiguous $\st$s above the critical $\ld$.
\item[2.] To calculate the probability of $L_{\mu}$, condition on the value of another series
of random variables, $\Psi_{\mu}$: the number of $\ld$s \emph{initially}
between the critical $\ld$ and the $\mu+1$th lowest $\st$.
\item[3.] Express the $\Psi_{\mu}$-conditioned probability of $L_{\mu}$ in terms of
the limit of the fraction of $\st$s near the bottom of a reordered thread, and
another probability, $\Pr[F_{\mu}|\Psi_{\mu} = q]$: the chance of $q$ $\ld$s all reordering
out of a region of \emph{at least} $\mu$ $\st$s.
\item[4.] To estimate $\Pr[F_{\mu}|\Psi_{\mu} = q]$, condition on a new random variable, $\Delta$:
the sum, over $\st$s, of the number of $\ld$s below each $\st$.
Express the probability of $\Delta$ in terms of the weighted sum of
several integer \emph{partition numbers}, and estimate these by a simple lower
bound.
\item[5.] After combining the above elements to bound the probability of
$L_{\mu}$, lower bound an ugly term of this expression by its value at
$\mu = 1$, checking via the derivative that this term is increasing
in $\mu$.
\item[6.] Use the lower bound on the probability of $L_{\mu}$ to finally lower
bound the probability of a given window size. To achieve an upper bound,
calculate the total probability not attributed to some $L_{\mu}$ in the
lower bound, and attribute it to the worst-possible case.
\squishend

We now move on to execute this plan in detail.

\vspace{2mm}\noindent\textbf{Step 1---Number of contiguous \boldmath{$\st$}s
  above the critical \boldmath{$\ld$}: } Recall that $\prgm_0$
($\prgm_{\plen+2}$) denotes the initial (final) instruction order, and that
$\prgm_{\plen}$ refers to the instruction order just \emph{before} the
critical load is settled. For convenience, we define the following basic
random events. Let $S_{\ld,i}(j)$ be the event that after the $j$th
instruction of $S_i$ is a $\ld$.
Furthermore, we define $S_{\ld,i}(j,k) = \bigwedge_{\ell=j}^k S_{\ld,i}(\ell)$
as the event that the entire contiguous range from $j$ to $k$ in $S_i$
consists of $\ld$s. $S_{\st,i}(j)$ and $S_{\st,i}(j,k)$ are
defined accordingly.

For $\mu \in \mathbb{N}$, we define $L_{\mu}$ as
the event that in $\prgm_{\plen}$, there are exactly $\mu$ $\st$
operations immediately preceding the critical $\ld$. In other words,
\begin{equation*}
L_{\mu} = S_{\ld,\plen}(\plen-\mu) \wedge
S_{\st,\plen}(\plen-\mu+1,\plen).
\end{equation*}
The critical $\ld$ may only move $\gamma$ positions if there are at
least $\gamma$ contiguous $\st$ operations above it. Hence for any
$\gamma$, we have
\begin{equation*}
\Pr[B_{\gamma}] = \sum_{\mu = \gamma}^{\infty}
\Pr[B_{\gamma} | L_{\mu}] \cdot
\Pr[L_{\mu}].\label{eq:vulnexpr}
\end{equation*}

Deriving $\Pr[B_{\gamma} | L_{\mu}]$ is straightforward. If $\mu = \gamma$,
we have $\Pr[B_{\gamma} | L_{\gamma}] = 2^{-\gamma}$, as the critical
$\ld$ must pass all $\gamma$ $\st$s. After that, it stops because the
next instruction is a $\ld$. For $\mu > \gamma$, we have
$\Pr[B_{\gamma} | L_{\mu}] = 2^{-(\gamma+1)}$, because the
instruction above the $\gamma$th $\st$ is also a $\st$. Hence there is
only a $1/2$ probability of the reordering completing when it reaches
that point.

It remains to derive bounds for $\Pr[L_{\mu}]$ for all $\mu$.
This is the primary technical lemma of the proof.

\begin{lemma}\label{lm:boundonlmu}
For any $\mu > 0$, $\Pr[L_{\mu}] \geq \frac{4}{7} \cdot 2^{-\mu}$.
Moreover, $\Pr[L_{0}] = 1/3$ exactly.
\end{lemma}

\begin{proof}
  We will approach this lemma by asking (1) how many $\ld$s are
  interspersed among the first $\mu$ $\st$s above the critical $\ld$,
  and (2) what is the probability that all of those $\ld$s settle such
  that we are left with $\mu$ contiguous $\st$s above the critical
  $\ld$. Because $\st$s cannot settle past $\ld$s in this model,
  nothing happens during rounds in which a $\st$ can move; the
  technical difficulty arises in the motion of the $\ld$s.

  \vspace{2mm}\noindent\textbf{Step 2---Number of interspersed
    \boldmath{$\ld$}s:} In the initial program order $\prgm_0$, let
  $\Phi_{\mu}$ refer to the position of the $\mu$th-lowest
  non-critical $\st$. Formally,
  \begin{equation*}
    \Phi_{\mu} = \min \{i : | \{j \geq i : S_{\st,0} (j) \} | = \mu + 1 \}.
  \end{equation*}
  Furthermore, let $\Psi_\mu$ refer to the number of $\ld$ operations
  above the critical $\ld$ but below the $\mu$th-lowest non-critical
  $\st$.  That is,
  \begin{equation*}
  \Psi_{\mu} = m+1 - \mu - \Phi_{\mu}.
  \end{equation*}
  Note that as the program length goes to infinity, the probability
  that such a $\Phi_{\mu}$ and $\Psi_\mu$ exist goes to 1.  Now we can
  express $\Pr[L_\mu]$ as
  \begin{equation}
    \Pr[L_{\mu}] = \sum_{q=0}^{\infty} \Pr[L_{\mu} | \Psi_{\mu} = q]
    \cdot \Pr[\Psi_{\mu} = q].\label{eq:lmunew}
  \end{equation}
  We have $\Pr[\Psi_{\mu} = q] = 2^{-\mu} 2^{-q} \binom{\mu+ q - 1}{q}$
  because there are $\binom{\mu + q - 1}{q}$ ways to build a
  string of $\mu$ $\st$s and $q$ $\ld$s such that the top instruction
  is a $\st$.

  \vspace{2mm}\noindent\textbf{Step 3---Probability of interspersed
    \boldmath{$\ld$}s settling out:} The difficult part of bound
  (\ref{eq:lmunew}) is $\Pr[L_{\mu} | \Psi_{\mu} = q]$. This is the
  probability that
\begin{enumerate}[(A)]
\item All $q$ $\ld$s between the $\st$ at $\Phi_{\mu}$ and the critical $\ld$ settle
up until they pass the $\st$ at $\Phi_{\mu}$,
\item but do not settle so far that the settled instruction above the $\st$ at
$\Phi_{\mu}$ is another $\st$. \label{enum:settletoofar}
\end{enumerate}
(\ref{enum:settletoofar}) is due to the fact that $L_{\mu}$ specifies
that there be \emph{exactly} $\mu$ $\st$s above the critical
$\ld$. The probability of (\ref{enum:settletoofar}) relies on the
instruction directly above $\Phi_{\mu}$ in $S_{\Phi_{\mu}-1}$. If it is a $\ld$,
then (\ref{enum:settletoofar}) holds automatically, since all the
$\ld$s must stop settling. However, if it is a $\st$, then
(\ref{enum:settletoofar}) only holds if not all of the q $\ld$s that
have passed the $\st$ at $\Phi_{\mu}$ also pass the next-highest
$\st$. Hence this is the first property on which we condition.
\begin{align*}
\Pr[L_{\mu} | \Psi_{\mu} = q] &= \Pr[L_{\mu} \wedge S_{LD,\Phi_{\mu}-1}(\Phi_{\mu}-1) | \Psi_{\mu} = q] \\
&+ \Pr[L_{\mu} \wedge S_{ST,\Phi_{\mu}-1}(\Phi_{\mu}-1) | \Psi_{\mu} = q].
\end{align*}
By Bayes' Law,
\begin{align*}
& \Pr[L_{\mu} \wedge S_{LD,\Phi_{\mu}-1}(\Phi_{\mu}-1) | \Psi_{\mu} = q] \\
&\quad = \Pr[S_{LD,\Phi_{\mu}-1}(\Phi_{\mu}-1) | \Psi_{\mu} = q] \\
&\quad\quad \cdot \Pr[L_{\mu} | S_{LD,\Phi_{\mu}-1}(\Phi_{\mu}-1) \wedge \Psi_{\mu} = q].
\end{align*}
We first consider the latter term.
Because the final instruction that settles above $\Phi_{\mu}$ will be a $\ld$
under these conditions, this depends only on the bottom $\mu$ instructions settled above
the critical $\ld$ being $\st$s.
For shorthand, let
\begin{equation*}
F_{\mu} = S_{ST,m}(m-\mu+1,m).
\end{equation*}
Then
\begin{equation*}
\Pr[L_{\mu} | S_{LD,\Phi_{\mu}-1}(\Phi_{\mu}-1) \wedge \Psi_{\mu} = q]
= \Pr[F_{\mu} |  \Psi_{\mu} = q].
\end{equation*}
In contrast, for $L_{\mu}$ to hold given $S_{ST,\Phi_{\mu}-1}(\Phi_{\mu}-1)$,
it does not suffice for the $q$ $\ld$s to move past $\Phi_{\mu}$. They must
also not all settle past the next highest instruction. They do so with probability
$2^{-q}$. Hence
\begin{multline*}
\Pr[L_{\mu} | S_{ST,\Phi_{\mu}-1}(\Phi_{\mu}-1) \wedge \Psi_{\mu} = q] = \\
\Pr[F_{\mu} |  \Psi_{\mu} = q] \cdot (1- 2^{-q}).
\end{multline*}
Putting these expressions together, we find that
\begin{align*}
& \Pr[L_{\mu} | \Psi_{\mu} = q] \\
&= \Pr[F_{\mu} |  \Psi_{\mu} = q] \cdot \Pr[S_{LD,\Phi_{\mu}-1}(\Phi_{\mu}-1)] \\
&\quad + \Pr[F_{\mu} |  \Psi_{\mu} = q] \cdot \Pr[S_{ST,\Phi_{\mu}-1}(\Phi_{\mu}-1)] \cdot (1- 2^{-q}) \\
&= \Pr[F_{\mu} |  \Psi_{\mu} = q] \cdot (1 - 2^{-q} \cdot (1 - \Pr[S_{ST,\Phi_{\mu}-1}(\Phi_{\mu}-1)])).
\end{align*}

  We first derive an exact value for $\Pr[S_{ST,i}(i)]$. Though it
  is difficult to determine the probability that a given instruction
  is a $\st$ in general, this particular value can be derived exactly
  through a recurrence relation.
    \begin{claim} \label{claim:0}
    \[
    \lim_{i \rightarrow \infty} \Pr[S_{ST,i}(i)] = 2/3.
    \]
    \end{claim}

    \begin{proof}
      After reordering stage $i$, instruction $i$ can be a ST in one
      of two ways. Either it can initially be a ST, (in which case it
      never reorders) or it can initially be a LD, the instruction
      above it can be settled as a ST, and the two can swap.  Hence
      \begin{equation*}
        \Pr[S_{ST,i}(i)] = \frac{1}{2} + \frac{1}{2} \cdot \Pr[S_{ST,i-1}(i-1)] \cdot \frac{1}{2}.
      \end{equation*}
      This is a recurrence relation of the form $X_i = b + aX_i$,
      which has the solution $X_i = \frac{b}{1-a} + a^{i-1} (X_1 -
      \frac{b}{1-a})$. Plugging in $X_1 = 1/2$, $a=1/4$, $b= 1/2$, we
      find
      \begin{align*}
        \Pr[S_{ST,i}(i)] &= \frac{1/2}{1-1/4} + (1/4)^{i-1} \left( 1/2 - \frac{1/2}{1-1/4} \right). \\
        &= 2/3 + (1/4)^{i-1} (1/2 - 2/3)
      \end{align*}
      The resulting probability is a function of $i$, but we are
      interested in the steady-state as the size of the program goes
      to infinity. Hence the second term falls out.
      \begin{equation*}
        \lim_{i \rightarrow \infty} \Pr[S_{ST,i}(i)] = 2/3. \qed
      \end{equation*}
    \end{proof}

    Now that we know the typical fraction of instructions near the
    bottom of the program that are $\st$s after reordering, we can
    derive a bound on $\Pr[F_{\mu}|\Psi_{\mu} = q]$.

\vspace{2mm}\noindent\textbf{Step 4---Estimating \boldmath{$\Pr[F_{\mu}|\Psi_{\mu} = q]$}:}

  \begin{claim}
    \label{claim:1}
    \begin{equation}
      \Pr[F_{\mu}|\Psi_{\mu} = q] \geq \frac{2^{-(q-1)} - 2^{-\mu q}}{\binom{\mu + q - 1}{q}}.\nonumber
    \end{equation}
  \end{claim}

  \begin{proof}
Everything in this proof is implicitly conditioned on the event
$\Psi_{\mu} = q$. Let the random variable
\begin{equation}
\Delta = \sum_{\Phi_{\mu} < i \leq \plen : \tau_{\ld,0}(i)} |\{\Phi_{\mu} \leq j < i
: \tau_{\st,0}(j)\}|\nonumber
\end{equation}
represent the total number of positions that $\ld$s from $\Phi_{\mu}$ to $\plen$
must move up, in order to leave a sequence of $\mu$ $\st$s immediately
above the critical $\ld$.  It must be that $\Delta \geq q$, because
at least instruction $\Phi_{\mu}$ is a $\st$, and $\Delta \leq \mu q$, because no
$\ld$ can be required to pass more than $\mu$ $\st$s. With this
definition, we may write $\Pr[F_{\mu}|\Psi_{\mu} = q] = \sum_{\delta =
  q}^{\mu q} \Pr[\Delta = \delta] \cdot 2^{-\delta}$.
The exact value of $\Pr[\Delta = \delta]$ can be stated formally, but
not in a closed form. Namely, let $\phi(x,y,z)$ be the number of
distinct multi-sets of $y$ positive integers summing to $x$, such that
each integer is at most $z$. This is a variant on the much-studied
\emph{partition number} of $x$. Then $\phi(\delta,q,\mu)$ is exactly
the number of arrangements of $q$ $\ld$s and $\mu$ $\st$s (beginning
with a $\st$) such that $\delta$ is the sum of the number of $\st$s
above each of the $\ld$s. (For each $\ld$, we simply select how many
$\st$s to place it below---the relative order of the $\ld$s is
immaterial.) There are $\binom{\plen + q - 1}{q}$ total arrangements of
$\ld$s and $\st$s beginning with a $\st$. Hence
\begin{equation*}
\Pr[\Delta = \delta] = \frac{\phi(\delta,q,\mu)}{\binom{\mu + q - 1}{q}},
\end{equation*}
and
\begin{equation*}
\Pr[F_{\mu}|\Psi_{\mu} = q] = \sum_{\delta
= q}^{\mu q} \frac{\phi(\delta,q,\mu)}{\binom{\mu + q - 1}{q}}
\cdot 2^{-\delta}.
\end{equation*}

Simple forms for $\phi(x,y,z)$ are not known. Asymptotic results
exist, but are not helpful here because the terms with small
parameters have the largest contributions. However, to achieve a good
bound it suffices to show that $\phi(\delta,q,\mu) \geq 1$ when $q
\leq \delta \leq \mu q$.  To show that a partition exists that
achieves any number in this range, consider the following
construction. Set $\delta \textrm{ mod } q$ of the integers to
$\ceil{\delta / q}$, and set the rest of the integers to
$\floor{\delta / q}$. We can set the integers this large, because
$\delta / q \leq (\mu q) / q = \mu$. Then the chosen integers sum to
$(\delta \textrm{ mod } q) \ceil{\delta / q} + (q - (\delta \textrm{
  mod } q)) \floor{\delta / q}$ which can be shown to be exactly
$\delta$. Hence we may write 
\begin{equation*}
\Pr[F_{\mu}|\Psi_{\mu} = q] \geq \frac{1}{\binom{\mu + q - 1}{q}} \sum_{\delta = q}^{\mu q} 2^{-\delta}
	= \frac{2^{-(q-1)} - 2^{-\mu q}}{\binom{\mu + q - 1}{q}}. \qed
\end{equation*}
\end{proof}

Having derived a bound for $\Pr[F_{\mu}|\Psi_{\mu} = q]$, we are now
in a position to conclude the proof of
Lemma~\ref{lm:boundonlmu}. First note that $\Pr[L_0] = 1/3$, by Claim
\ref{claim:0}.  For values of $\mu$ greater than 0, Claim
\ref{claim:1} will be the central tool in the proof, which is left to
\ifthenelse{\boolean{fullversion}} {Appendix~\ref{app:critproofs}.}
           {the full version of the paper.} \qed
\end{proof}
The remainder of the proof of Theorem~\ref{thm:crit-total}, steps 5
and 6, is deferred to \ifthenelse{\boolean{fullversion}}
{Appendix~\ref{app:critproofs}.}  {the full version of the paper.}
\qed
\end{proof}

\section{Shift Process}
\label{sec:shift}

Here we discuss the next component of our analysis: a ``shift process'' meant to
capture the interleaving of reordered threads. We refer the reader back to the
definition in Section \ref{sec:interleaving-model}. This process is where the
critical windows derived from the reordering process come into effect.

In the analysis that follows, we assume that each critical window's
shift is distributed geometrically, representing the intuition that
threads are exponentially less likely to execute at progressively
increasing offsets from one another.  Let $\lvalvec = (\lval_{1},
\lval_{2}, \ldots, \lval_{n}) \in \mathbb{N}^n$ be a sequence of
integral ``segment lengths.'' In subsequent sections, $\lval_{k}$ will
be used to represent the length of the critical window of thread
$T_k$.  We define a shift process on $\lvalvec$ as follows. Consider
$n$ segments of the line, of lengths $\lval_{1}, \lval_{2}, \ldots,
\lval_{n}$, and let the starting point of each segment be shifted up
from 0 by an i.i.d.\ positive random variable $s_i$. We are interested
in the probability that the resulting set of shifted segments is
non-overlapping. In other words, we would like to bound
$\Pr[A(\lvalvec)]$, where $A(\lvalvec)$ is the event that $\forall i
\neq j \in \{1, 2, \ldots, n\}$, we have $[s_i, s_i + \lval_{i}] \cap
[s_j, s_j + \lval_j] = \emptyset$.

The following theorem states this probability precisely, and as such is not
particularly enlightening on its own. However, when the segment lengths
are random variables drawn from a well-understood distribution
(as they are in the case of reordered random threads), we will be able to
state the probability concisely.\vspace{-2mm}
\begin{thm} \label{thm:shift}
\begin{equation*}
\Pr[A(\lvalvec)] = \frac{2^{-(\binom{n+1}{2} - 1)}}{\prod_{i=1}^{n-1} (1-2^{-(n+1-i)})}
	\sum_{\sigma \in Sym_n} \prod_{i=1}^{n-1} 2^{-(n-i) \lval_{\sigma(i)}},\vspace{-1mm}
\end{equation*}
where $Sym_n$ is the symmetric group of degree $n$: the set of all permutations
on $n$ elements.
\end{thm}
The following corollary simplifies this expression:
%
\begin{corollary} \label{cor:shift}
For some $c(n) \in [2,4]$, 
\begin{equation*}
\Pr[A(\lvalvec)] =
	c(n) \cdot 2^{-\binom{n+1}{2}} \cdot
	\sum_{\sigma \in Sym_n} \prod_{i=1}^{n-1} 2^{-(n-i) \lval_{\sigma(i)}}.
\end{equation*} 
	In particular, $c(2) = \frac{8}{3}$ exactly. 
%

\end{corollary}
The proof of the corollary is in \ifthenelse{\boolean{fullversion}}
{Appendix \ref{app:shiftproofs}.}  {the full version of the paper.}
We now turn to the proof of the main theorem. The challenge is to
characterize the probability that the next segment is shifted to a
position disjoint from all previous segments. At first glance, it is
difficult to handle the huge and diverse set of legal placements
for a set of segments. Our key insight is to condition on the
\emph{relative order} of the magnitude of the shifts.  We then
consider the probability that each segment is disjoint from the
previous threads in this order. In so doing, we are able to exploit
the memorylessness of the geometric distribution. Let $t$ be an
arbitrary segment, and $t'$ be the segment immediately preceding it in
this order. To understand the distribution of disjoint placements for
$t$, we need only know the distribution of the origin of $t'$.  Then
by assuming that the segments are disjoint, we can infer that the
origin of $t$ is distributed according the origin of $t'$, plus the
length of $t'$, plus an independent geometric random variable.
\begin{proof} [(Theorem \ref{thm:shift})]
  Let $s_i$ be a geometric random variable with expectation 2 (i.e.,
  $s_i = k$ with probability $2^{-(k+1)} \ \forall k \in
  \mathbb{N}$). In order to analyze the probability of $A(\lvalvec)$,
  we will take the following steps. We will first condition on the
  ordering of the segments. Then for a given ordering, we will use the
  memorylessness of the shift variables to calculate the probability
  of each successive segment being disjoint from each previous.

For a permutation $\sigma$ on $\{1, 2, \ldots, n\}$, let $Y_{\sigma}$ be the event that for all $i$,
the $i$th largest shift occurs on segment $\sigma(i)$. That is,
$s_{\sigma(1)} \geq s_{\sigma(2)} \geq \cdots \geq s_{\sigma(n)}$.
Then $\Pr[A(\lvalvec)] = \sum_{\sigma \in Sym_n} \Pr[A(\lvalvec) \wedge Y_{\sigma}]$.

We now analyze $\Pr[A(\lvalvec) \wedge Y_{\sigma}]$. We will refer to this event by $A(\lvalvec, \sigma)$.
For all segments to be disjoint, it must be the case that each segment
begins after the \emph{end} of every segment that began before
it. $\sigma$ captures exactly the order in which segments begin. So
disjointness means that for all $i$, $j$ s.t. $\sigma(j) > \sigma(i)$,
segment $j$ begins after the end of segment $i$. Hence for each $i$,
we may condition on the shift of the segment with the $i$th largest
shift, and consider the probability that each segment with a smaller
shift follows its completion.
\begin{gather*}
  \Pr[A(\lvalvec, \sigma)] = \sum_{\ell_1 = 0}^{\infty} \Pr[A(\lvalvec, \sigma) \wedge s_{\sigma(1)} = \ell_1] \\
  \begin{align*}
    &= \!\sum_{\ell_1 = 0}^{\infty} \Pr[A(\lvalvec, \sigma) \wedge s_{\sigma(1)} = \ell_1
      \wedge \bigwedge_{i = 2}^n s_{\sigma(i)} \geq \ell_1 + \lval_{\sigma(1)}] \\
    &= \sum_{\ell_1 = 0}^{\infty} \Pr[A(\lvalvec, \sigma) | s_{\sigma(1)} = \ell_1
      \wedge \bigwedge_{i = 2}^n s_{\sigma(i)} \geq \ell_1 + \lval_{\sigma(1)}] \\
    & \quad \quad \quad \cdot \Pr[s_{\sigma(1)} = \ell_1]
    \cdot \prod_{i = 2}^n \Pr[s_{\sigma(i)} \geq \ell_1 \!+\! \lval_{\sigma(1)}].
  \end{align*}
\end{gather*}
The third equality is due to the independence of the shift variables.
Let $\rgam{i}$ refer to the restriction of $\lvalvec$ to the segment
indices with the $n-i+1$ smallest shifts (i.e., $\rgam{i} =
\lvalvec_{|[n] \setminus \bigcup_{j=i}^n \sigma(j)}$). Similarly, let
$\rsig{i}$ refer to the restriction of $\sigma$ to the $n-i+1$
smallest shifts (i.e., $\rsig{i} = \sigma_{|[n] \setminus [i-1]}$).
We define these structures so that we can express the disjointness
event in terms of a new disjointness event on a smaller set of
unconditioned segments.  In particular, let $A(\rgam{i}, \rsig{i})$ be
the disjointness event for an independent random shift process on
segments $\sigma(i), \sigma(i+1), \ldots, \sigma(n)$, with permutation
$\rsig{i}$ pointing to the new indices of these segments.  We will see
that we are permitted to condition on such a prior event, because of
the memoryless of the shift variables.

Conditioned on the first segment being disjoint from all the following segments,
we need only consider the event $A(\lvalvec^{2},\sigma)$.
Then due to the memorylessness of the shifts, we have
\begin{multline*}
\Pr[A(\rgam{i}, \rsig{i}) | s_{\rsig{i}(1)}
	= \ell_1 \wedge \bigwedge_{j = 2}^n s_{\rsig{i}(j)} \geq \ell_1 + \lval_{\rsig{i}(1)}] \\
	= \Pr[A(\rgam{i+1}, \rsig{i+1})
	| \bigwedge_{j = 2}^n s_{\rsig{i}(j)} \geq \ell_1 + \lval_{\rsig{i}(1)}] \\
= \Pr[A(\rgam{i+1}, \rsig{i+1}) | \bigwedge_{j = 2}^n s_{\rsig{i}(j)} \geq 0] = \Pr[A(\rgam{i+1}, \rsig{i+1})].
\end{multline*}
We now observe a simple recurrence relation that defines $\Pr[A(\rgam{i}, \rsig{i})]$.
\begin{multline*}
  \Pr[A(\rgam{i}, \rsig{i})] = \sum_{\ell_1 = 0}^{\infty} \Pr[A(\rgam{i+1}, \rsig{i+1})]
  \cdot \Pr[s_{\rsig{i}(1)} = \ell_1]\\
  \cdot \prod_{j=i+1}^n \Pr[s_{\rsig{i}(j)} \geq \ell_1 + \lval_{\rsig{i}(1)}]
\end{multline*}
\begin{align*}
&= \sum_{\ell_1 = 0}^{\infty} \Pr[A(\rgam{i+1}, \rsig{i+1})]
\cdot \frac{1}{2} 2^{-\ell_1} \cdot \prod_{j=i+1}^n \frac{1}{2} \cdot 2^{-(\ell_1 + \lval_{\rsig{i}(1)})} \\
&= \sum_{\ell_1 = 0}^{\infty} \Pr[A(\rgam{i+1}, \rsig{i+1})]
\cdot 2^{-(\ell_1 + 1 + (n-i) (\ell_1 + \lval_{\rsig{i}(1)} + 1))} \\
&= 2^{-1 + (n-i) (\lval_{\rsig{i}(1)} + 1)}
\cdot \Pr[A(\rgam{i+1}, \rsig{i+1})] \sum_{\ell_1 = 0}^{\infty} (2^{-(n-i+1)})^{\ell_1} \\
&= \frac{2^{-1 + (n-i) (\lval_{\rsig{i}(1)} + 1)}}{1 - 2^{-(n-i+1)}}
\cdot \Pr[A(\rgam{i+1}, \rsig{i+1})].
\end{align*}

Moreover, it is clear that $\Pr[A(\lvalvec^{n}, \sigma^{n})] = 1.$ Then noting that
$\rsig{i}(1) = \sigma(i)$, the solution is trivial:
\begin{align*}
\Pr[A(\lvalvec^{1}, \sigma^{1})] &=
\prod_{i=1}^{n-1} \frac{2^{- (n+1-i) - (n-i) \lval_{\sigma(i)}}}{1 - 2^{-(n+1-i)}} \\
&=\frac{2^{-(\binom{n+1}{2} - 1)}}{\prod_{i=1}^{n-1} (1-2^{-(n+1-i)})}
		\cdot \prod_{i=1}^{n-1} 2^{-(n-i) \lval_{\sigma(i)}}.
\end{align*}
Finally, plugging these terms into the overall probability of
disjointness yields the expression in the theorem. We will use this
expression in the next section to calculate the probability of bug
manifestation.
\end{proof}
%
%
%


%
%
%

\section{Joining the Models}
\label{sec:together}

We have now described the two fundamental random processes of our
work. Though the two are interesting in isolation, it is by combining
them that we will achieve our overall goal: to characterize the
probability of the canonical data race manifesting, under various
memory models.

Our first observation is to note that Corollary \ref{cor:shift} can be
further simplified, provided the segment lengths are drawn from a
distribution with a very weak condition.

\begin{thm} \label{thm:identical} Let $\lvarvec = \lvar_1, \ldots,
  \lvar_n$ be a distribution over segment lengths, drawn from
  $\mathbb{N}^n$. Assume that the marginal distribution of each
  segment length is identical (i.e., $\lvar_i \sim \lvar_j \ \forall \
  i \neq j$); they needn't be independent. Then all permutations of
  segment shifts are equivalent, and
\begin{equation*}
\Pr[A(\lvarvec)] = c(n) \cdot 2^{-\binom{n+1}{2}} \cdot n!
	\cdot \E_{\lvarvec} [\prod_{i=1}^{n-1} 2^{-i \lvar_{i}}].
\end{equation*}
\end{thm}
The proof is given in
\ifthenelse{\boolean{fullversion}}
{Appendix~\ref{app:finalthms}.}
{the full paper.}
Because the identicality condition holds for the critical window size,
the theorem gives an indication of how it is that we can analyze the
overall bug manifestation concretely.  Recall that the process of
Section~\ref{sec:critwindow} generates a uniformly random program of
$\st$s and $\ld$s, then randomly ``settles'' each instruction in turn,
according to the rules of the memory model. The process of
Section~\ref{sec:shift} applies a random ``shift'' to a series of line
segments, the key event for which is the mutual disjointness of all
the segments. We now combine these two processes by letting the line
segment lengths of the shift process be distributed as the critical
window size of the settling process.  An important subtlety is that
\emph{we generate a single initial random program, then independently
  reorder n copies of this program}. Though this makes the analysis
more complex, it adds a degree of realism: with $n$ identical threads,
it is more natural that the same data race would be present in the
same position of every pair of threads.  The following two theorems
summarize our key results.
%
\begin{thm} \label{thm:2threads}
  For $n=2$ threads, the probability that the canonical data race does
  not manifest is the following, in each of the three main models.
  \begin{center}
    \begin{tabular}{lc}
      \emph{\textbf{Sequential Consistency:}} & $\Pr[A] \approx 0.1666$ \\
      \emph{\textbf{Total Store Order:}}\footnotemark & $0.1369 > \Pr[A] > 0.1315$ \\
      \emph{\textbf{Weak Ordering:}} & $\Pr[A] \approx 0.1296$ \\
    \end{tabular}
  \end{center}
\end{thm}
\begin{thm} \label{thm:nthreads}
As $n$ grows, the probability of successful execution is identical in
\emph{all} models, up to lower order terms in the exponent. In particular,
$\Pr[A] = e^{-n^2(1+o(1))}$.
\end{thm}
\footnotetext{A very similar analysis achieves a similar result for
  Partial Store Order (PSO).  We omit the result for brevity.}  The
first tightly bounds the probability of successful execution for the
case of $n=2$ threads; the second gives an asymptotic bound on this
probability for large $n$. We leave the proofs of these theorems to
\ifthenelse{\boolean{fullversion}}
{Appendix~\ref{app:finalthms}.}
{the full paper.}
Both proofs are rather technical and build upon the theorems of the previous two sections. The only surprising
observation necessary is that, when lower bounding
a certain expectation over the critical window for $n$ threads, it suffices
to use only a single term of this expectation. Doing so
achieves the asymptotic behavior we seek.

\noindent \textbf{Key Observations: }Interpreting
Theorems~\ref{thm:2threads} and~\ref{thm:nthreads} yields remarkable
insights. Though the case of $n=2$ substantively distinguishes the
memory models, we find that as $n$ grows, the probability in all
memory models approaches the same value, up to lower order terms in
the exponent.  This dichotomy is a fundamental take-away for informing
computer architecture decisions. Though the use of weaker memory
models does increase the risk of program error, as the number of
threads grows this risk grows negligibly compared to growth of risk of
error in even sequential consistency. This is of particular importance
given the trends towards ever larger multicores that enable more and
more concurrent threads.

\section{Discussion}
\label{sec:discussion}




\noindent{\bf Limitations and possible extensions:} Our analysis
assumes that the program consists solely of loads and stores, when
real programs include synchronization, arithmetic, etc. These
instructions can affect the timing of the program, introduce data
dependencies that limit reordering, or disallow certain types of
reorderings. An important item for future work is to include acquire/release
fences, which are necessary to simulate memory models such as Release
Consistency~\cite{gharachorloo90}. These fences act as one-way
barriers, allowing instructions to reorder into, but not out of, a
critical section. This behavior can be easily modeled using settling
(\S\ref{sec:reordering_model}). Fences make concurrency bugs less
likely to manifest, as programs with fences have fewer legal
reorderings. However, we conjecture that adding fences will not significantly
change the main conclusions derived in this paper.

\noindent{\bf Optimized implementations of SC: }Our model
of Sequential Consistency assumes a relatively simple implementation
wherein each processor executes only one memory instruction at a
time. Many SC implementations use aggressive optimizations such as
speculative execution to compete with the performance of weaker memory
models \cite{gharachorloo91a, gniady99, ceze07}.
We do not consider this simplifying assumption to be a weakness of our
model; rather, we believe our results about weak memory models can be
extended to address optimizing implementations of strong memory
models. In other words, concurrency bugs are more likely to manifest
in an implementation of SC that uses aggressive reordering than in a
simple (and slow) implementation.

 \noindent{\bf Generality of Results: }In this paper, we propose and study one specific probabilistic process to model program execution and thread interleaving. Clearly, there are other plausible models that can be studied. Our intuition is that the results in this paper have a certain robustness with regard to changes to the parameters in our models as well as to changes in the model. However, future work is required to formally validate this conjecture.

\section{Conclusion}

With the ubiquity of multicore systems and the trend towards integrating every more cores on a single chip, multiprocessor programmability has become one of the key challenges in computer science. Even with improvements in
programmability, we are likely to see an increase in software defects,
given the inherent difficulty of concurrent programming.
Memory consistency models are at the center of the programmability
discussion, since they determine the memory access semantics of
parallel programs. The debate over memory models has historically
revolved around the trade-offs between programmability, performance
and complexity. In this paper we bring a new axis to this discussion:
{\em software reliability}. We study an analytical model and show that concurrency bugs are indeed more
likely to manifest themselves in relaxed memory models, but surprisingly, that as the number of parallel threads increases, the difference between harsh and weak memory models diminishes. The latter observation can have important consequences on system designers when developing new memory models.




\bibliographystyle{abbrv}
\bibliography{bibliography}

\ifthenelse{\boolean{fullversion}}
{\newpage

\appendix

\section{Model Definition}

\subsection{Initial Program Order}
\label{app:initialorder}
The initial program order $S_0$ consists of $n+2$ instructions:
\begin{equation*}
x_1, x_2, \ldots, x_n, \ld~X, \st~X
\end{equation*}
where for $1 \le i \le n$, $x_i$ has type $\tau(x_i) = \st$ with
probability $p$, and type $\tau(x_i) = \ld$ with probability $(1 -
p)$. Each $x_i$ accesses a location $X_i$ such that $X_i = X_j$ only
if $i = j$, and $X_i \ne X$. For the purposes of defining the model we
assume that $n$ is finite, but in the analysis it is useful to
approximate a very long program by letting $n \to \infty$.

\subsection{Definition of Settling Process}
\label{app:modeldef}

We model instruction reordering as a random process consisting of
$n+2$ rounds. This process produces a permutation of $S_0$ which we
call $S_{n+2}$; round $i$ produces the intermediate permutation
$S_i$. During round $i$, instruction $x_{i}$ is inserted into the
permuted ordering of instructions $x_1$ through $x_{i-1}$. We decide
where to insert instruction $x_{i}$ by repeatedly swapping $x_{i}$
with the instruction directly before it.  Each swap succeeds with
probability $\rho_{\tau_1,\tau_2}$, where $\tau_1$ is the type of the
instruction directly before $x_{i}$'s current location and $\tau_2$ is
the instruction type of $x_{i}$, and fails with probability $1 -
\rho_{\tau_1,\tau_2}$. $\rho_{\tau_1,\tau_2}$ is always either 0 or
$s$, depending on the memory model. The single exception is for the
critical $\ld$ and $\st$, $x_{plen+1}$ and $x_{\plen+2}$. If $x_{\plen+2}$
ever tries to swap with $x_{\plen+1}$, it automatically fails, because
they access the same memory address. The round completes when a swap
fails occur or $x_i$ reaches position 1. This recursive random process
is called \emph{settling}.

Let $\permi{i}{j}$ be a function from positions in $S_0$ to positions
in $S_i$. (Note that $\permi{0}{j} = j$ for all $j$.)
We formally define the insertion point of instruction $x_i$ using the
probability distribution $\beta_i$.

\begin{defn}
  \label{def:beta}
  Given the intermediate permutation $S_{i-1}$, we define a
  probability distribution $\beta_{i,k}$ as follows:
  \begin{itemize}
  \item If $k = 1$, $1$.
  \item Else let $j = \permiinv{i-1}{k-1}$ and let $q = \rho_{\tau(x_j),\tau(x_i)}$.
    \begin{itemize}
    \item $k$ with probability $1-q$
    \item Draw from $\beta_{i,k-1}$ with probability $q$.
    \end{itemize}
  \end{itemize}
  We also define $\beta_i$ to be $\beta_{i,i}$.
\end{defn}

$\beta_i$ describes the distribution of possible positions for
instruction $i$ after round $i$ of settling. $\beta_{i,k}$ describes
the distribution of the possible positions for instruction $i$ given
that $i$ moves up at least as high as position $k$.

The result of round $i$ is the permutation $\pi_i$, in which the
instructions following $x_{i}$'s new location are each pushed down by
one, and the instructions preceding $x_{i}$'s new location do not move.

\begin{defn}
  \label{def:pi}
  Recall that $\pi_i$ is a function mapping positions in $S_0$ to
  $S_i$.  Given permutation $S_{i-1}$, we draw $k$ from $\beta_i$ and
  construct the permutation $S_i$ as follows:
  \begin{align*}
    \permi{i}{i} &= k \\
    \permi{i}{j} &= \permi{i-1}{j} \textrm{~~for~~} \permi{i-1}{j} < k \\
    \permi{i}{j} &= \permi{i-1}{j} + 1 \textrm{~~for~~} \permi{i-1}{j} \ge k
  \end{align*}
\end{defn}

We use definitions \ref{def:beta} and \ref{def:pi} to get a
probability distribution over permutations of $S_0$.
We refer to the final permutation $\pi_{\plen+2}$ as $\pi$.


\subsection{Definition of Interleaving Model}
\label{sec:interleavingdef}

Formally, the thread interleaving model is defined as follows.
Let threads $T_1, \ldots, T_n$ be $n$ identical threads,
distributed as described in \ref{app:initialorder}.

We allow each initially-identical thread to reorder independently,
using the settling process of \ref{app:modeldef}.
We refer to the final permutation $\pi_{\plen+2}$ of thread $T_k$ as $\pi^{(k)}$.
Define the ``critical window'' $W_k$ for a reordered thread $T_k$ to be the
set of indices (inclusively) between the settled positions of
the critical instructions. E.g.,
\begin{multline*}
W_k = \{\pi^{(k)}(\plen+1), \pi^{(k)}(\plen+1) + 1, \ldots,\\
\pi^{(k)}(\plen+2) - 1, \pi^{(k)}(\plen+2) \}.
\end{multline*}
Finally, we independently allow each thread to ``shift up'' with
respect to one another.

For each $k$, we allow thread $T_k$ to shift exactly $i$ positions
up with probability $2^{-(i+1)}$. Observe that $\sum_{i=0}^{\infty}
2^{-(i+1)} = 1$, so as $n \to \infty$, this gives a
probability distribution over the positions of each instruction
in each thread. Let $\eta_k$ be the shift of $T_k$.

We then say that the bug \emph{manifests} if there
exist $k \neq \ell$ such that the critical windows of reordered $T_k$
and $T_{\ell}$ overlap whatsoever. In other words, define the
bug non-manifestation event $A$ by
\begin{multline*}
A = \neg \exists k \neq \ell : (\pi^{(\ell)}(\plen+2)-\eta_{\ell} \geq \pi^{(k)}(\plen+2)-\eta_k) \\
  \wedge (\pi^{(\ell)}(\plen+2)-\eta_{\ell} \leq \pi^{(k)}(\plen+1)-\eta_k).
\end{multline*}
(Note that for any overlapping pair of ranges, the bottom of one
of the two windows will necessarily overlap with the other window.)

Expressed alternately, let $W_k'$ be the shifted window \\
$\{\pi^{(k)}(\plen+1) - \eta_k, \pi^{(k)}(\plen+1) + 1 - \eta_k, \ldots,
\pi^{(k)}(\plen+2) - 1 - \eta_k, \pi^{(k)}(\plen+2) - \eta_k \}$.
Then
\[
A = \neg \exists k \neq \ell : W_k' \cap W_{\ell}' \neq \emptyset.
\]

\begin{figure*}[htb]
  \begin{center}
    \begin{center}
      \includegraphics[scale=0.5]{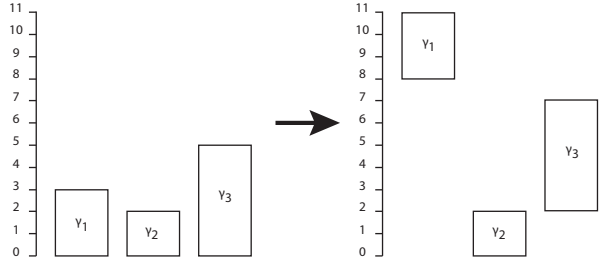}
    \end{center}
    \caption{An instantiation of the shift process. Three segments of
      lengths $\bar{\gamma} = (\gamma_1, \gamma_2, \gamma_3) =
      (3,2,5)$ are independently shifted. This particular shift occurs
      with probability $2^{-8-1} \cdot 2^{-0-1} \cdot 2^{-2-1} =
      2^{-13}$. The disjointness event $A(\bar{\gamma})$ does indeed
      hold here. \label{fig:shift}}
  \end{center}
\end{figure*}

\section{Proofs}

\subsection{Proofs for Theorem \ref{thm:crit}: Critical Window Growth}
\label{app:critproofs}


In this section, we finish up the proofs for
Lemma \ref{lm:boundonlmu}, and the Total Store Order case of
Theorem \ref{thm:crit}.



\begin{proof}[(Remainder of Lemma \ref{lm:boundonlmu})]
  \begin{align*}
    \Pr[L_{\mu}] &= \Pr[L_{\mu}|\Psi_{\mu} = q] \cdot \Pr[\Psi_{\mu} = q] \\
    &= \sum_{q=0}^{\infty} 2^{-\mu} 2^{-q} \binom{\mu + q - 1}{q} \Pr[F_{\mu} | \Psi_{\mu} = q] \\
    & \quad\quad\quad \cdot \left( 1 - (1-g(\mu,q)) \cdot \Pr[S_{\st,\Phi_{\mu}-1}(\Phi_{\mu}-1)] \right) \\
    &= \sum_{q=0}^{\infty}2^{-\mu} 2^{-q} \binom{\mu + q - 1}{q}
    \Pr[F_{\mu} | \Psi_{\mu} = q] \\ 
    & \quad\quad\quad \cdot \left( 1 - 2^{-q} \cdot \frac{2}{3} \right)
  \end{align*}
  By Claim~\ref{claim:1}:
  \begin{align*}
    &\geq 2^{-\mu} \sum_{q=0}^{\infty} 2^{-q} \binom{\mu + q - 1}{q}
    \cdot \frac{2^{-(q-1)} - 2^{-\mu q}}{\binom{\mu + q - 1}{q}} \left(1 - 2^{-q} \frac{2}{3}\right) \\
    &= 2^{-\mu} \sum_{q=0}^{\infty} \bigg(2^{-(2q-1)} - \frac{2}{3} 2^{-(3q-1)}
    - 2^{-\mu q - q} + \frac{2}{3} 2^{-\mu q - 2 q}\bigg) \\
    &=  2^{-\mu} \left(\frac{2}{1 - \frac{1}{4}} - \frac{4/3}{1 - \frac{1}{8}}
    - \frac{1}{1 - \frac{1}{2^{\mu+1}}} + \frac{2/3}{1 - \frac{1}{2^{\mu+2}}} \right) \\
    &=  2^{-\mu}\left(\frac{8}{7}
    - \frac{1}{1 - \frac{1}{2^{\mu+1}}} + \frac{2/3}{1 - \frac{1}{2^{\mu+2}}} \right)
  \end{align*}

  The above expression is difficult to work with, so we give a simpler
  lower bound that holds for all $\mu \geq 1$. Let $h(\mu)$ be the
  parenthesized expression above (such that the bound is $\Pr[L_{\mu}]
  \geq 2^{-\mu} \cdot h(\mu)$).

  We differentiate $h(\mu)$ and show that it is increasing, and
  compute a small value explicitly, so that we may lower bound all
  higher values by this value.

  We will show that for all  $\mu \geq 1$,
  \[
  h(\mu) \geq 4/7
  \]

  We first note that	
  \begin{align*}
    h(1) &= 8/7 - \frac{1}{1 - 1/4} - \frac{2}{3} \cdot \frac{1}{1 - 1/8} \\
    &= 8/7 - 4/3 + 16/21 \\
    &= 4/7.
  \end{align*}

  We now show that $h(\mu)$ is increasing in $\mu$.  The function
  $h(\mu)$ is defined as $h(\mu) = \frac{8}{7} - (1 - 2^{-(\mu +
    1)})^{-1} + \frac{2}{3} \cdot (1 - 2^{-(\mu+2)})^{-2}$.  To see
  that $h(\mu)$ is increasing, we differentiate w.r.t. $\mu$.
  \begin{equation*}
    \frac{d}{d \mu} h(\mu) = \frac{2^{-(\mu+1)} (\ln 2)}{(1-2^{-(\mu+1)})^2}
    - \frac{2}{3} \cdot \frac{2^{-(\mu+2)} (\ln 2)}{(1-2^{-(\mu+2)})^2}
  \end{equation*}
  This expression is positive when
  \begin{equation*}
    \frac{1}{(1 - 2^{-(\mu+1)})^2} > \frac{4}{3 (1-2^{-(\mu+2)})^2},
  \end{equation*}
  e.g. $\frac{1 - 2^{-(\mu+2)}}{1 - 2^{-(\mu+1)}} >
  \sqrt{\frac{4}{3}}$.  This holds for all $\mu \geq 0$. Hence
  $h(\mu)$ is descreasing for non-negative $\mu$.
\end{proof}


\begin{proof}[(Remainder of Theorem~\ref{thm:crit})]
  ~\\
  It will be useful to calculate the total \emph{slack} in our
  probability bounds. Let $\Pr_{\ell}[L_{\mu}]$ be the lower bound for
  $\Pr[L_{\mu}]$ computed in this section, and $\Pr_{r}[L_{\mu}] =
  \Pr[L_{\mu}] - \Pr_{\ell}[L_{\mu}]$ be the remainder. The
  ``missing'' probability $R = \sum_{\mu=0}^{\infty} \Pr_r[L_{\mu}]$
  is computed below.
  \begin{claim} \label{claim:rbound}
    \begin{equation*}
      R = 2/21.
    \end{equation*}
  \end{claim}
  \begin{proof}
    Note that $\Pr_{r}[L_{\mu}]$ is nonnegative for all $\mu$, because
    $\Pr_{r}[L_{\mu}]$ is a lower bound. $L_{\mu}$ must hold for
    exactly one $\mu$, hence $\sum_{\mu=0}^{\infty} \Pr[L_{\mu}] = 1$.
    \begin{align*}
      R &= 1 - \sum_{\mu = 0}^{\infty} \Pr_{\ell}[L_{\mu}] \\
      &= 1 - \Pr_{\ell}[L_0]  - \sum_{\mu = 1}^{\infty} h(1) \cdot 2^{-\mu} \\
      &= 1 - 1/3 - 4/7 \cdot 1 \\
      &= 2/3 - 4/7 \\
      &= 2/21. \qed
    \end{align*}
  \end{proof}

  Now, using the derivation of $\Pr[L_0]$ from Lemma
  \ref{lm:boundonlmu},
  \begin{align*}
    \Pr[B_0] &= \Pr[B_0|L_0] \cdot \Pr[L_0] + \Pr[B_0|\neg L_0] \cdot \Pr[\neg L_0] \\
    &= 1 \cdot (1/3) + (1/2) \cdot (2/3) \\
    &= 2/3.
  \end{align*}

  Moving on to the case of $\gamma \geq 1$, we rewrite
  $\Pr[B_{\gamma}]$ as
  \begin{equation*}
    \Pr[B_{\gamma}] = \sum_{\mu = \gamma}^{\infty} \Pr[B_{\gamma} |
      L_{\mu}]  \cdot \Pr_{\ell}[L_{\mu}]
    + \sum_{\mu = \gamma}^{\infty} \Pr[B_{\gamma} | L_{\mu}]  \cdot \Pr_r[L_{\mu}].
  \end{equation*}
  We compute the first sum exactly, and provide upper and lower bounds
  for the second sum, in order to upper and lower bound
  $\Pr[B_{\gamma}]$.

  First we compute the value of $\sum_{\mu = \gamma}^{\infty}
  \Pr[B_{\gamma} | L_{\mu}] \cdot \Pr_{\ell}[L_{\mu}]$ exactly:
  \begin{align*}
    \sum_{\mu = \gamma}^{\infty} &\Pr[B_{\gamma} | L_{\mu}]  \cdot \Pr_{\ell}[L_{\mu}] \\
    &= \Pr[B_{\gamma} | L_{\gamma}]  \cdot \Pr_{\ell}[L_{\gamma}]
    + \sum_{\mu = \gamma+1}^{\infty} \Pr[B_{\gamma} | L_{\mu}]  \cdot \Pr_{\ell}[L_{\mu}] \\
    &= 2^{-\gamma} h(1) 2^{-\gamma} + \sum_{\mu = \gamma+1}^{\infty} 2^{-\gamma} (1/2) h(1) 2^{-\mu} \\
    &= h(1) 2^{-\gamma} (2^{-\gamma} + \sum_{\mu = \gamma+1}^{\infty} (1/2) 2^{-\mu}) \\
    &= h(1) 2^{-\gamma} (2^{-\gamma} + (1/2) \frac{2^{-(\gamma+1)}}{1/2}) \\
    &= h(1) 2^{-\gamma} \cdot 3 \cdot 2^{-(\gamma+1)} \\
    &= 3 h(1) 2^{-(2\gamma + 1)} \\
    &= \frac{6}{7} \cdot 4^{-\gamma}.
  \end{align*}

  Next we upper bound $\sum_{\mu = \gamma}^{\infty} \Pr[B_{\gamma} |
    L_{\mu}] \cdot \Pr_r[L_{\mu}]$:
  \begin{align*}
    \sum_{\mu = \gamma}^{\infty} & \Pr[B_{\gamma} | L_{\mu}]  \cdot \Pr_r[L_{\mu}] \\
    &= 2^{-\gamma}  \cdot \Pr_r[L_{\gamma}] + \sum_{\mu = \gamma+1}^{\infty} 2^{-\gamma} (1/2) \cdot \Pr_r[L_{\mu}] \\
    &= 2^{-\gamma} \cdot (\Pr_r[L_{\gamma}] + (1/2) \sum_{\mu = \gamma+1}^{\infty} \Pr_r[L_{\mu}]).
  \end{align*}

  To upper bound the above expression, observe that
  \begin{equation*}
  \sum_{\mu=\gamma}^{\infty} \Pr_r[L_{\mu}] \leq R.
  \end{equation*}
  It is clear that allocating all of this probability mass to
  $L_{\gamma}$ maximizes the above expression:
  \begin{equation*}
    \sum_{\mu = \gamma}^{\infty} \Pr[B_{\gamma} | L_{\mu}] \cdot
    \Pr_r[L_{\mu}] \le 2^{-\gamma} \cdot (R + (1/2) \sum_{\mu =
      \gamma+1}^{\infty} 0) = R 2^{-\gamma}
  \end{equation*}
  We cannot ensure that $\sum_{\mu= \gamma}^{\infty}
  \Pr_r[L_{\gamma}]$ is positive for $\gamma > 0$, because all of $R$
  could be allocated to $\Pr_r[L_0]$. Hence the best lower bound we
  can expect here is $\sum_{\mu = \gamma}^{\infty} \Pr[B_{\gamma} |
    L_{\mu}] \cdot \Pr_r[L_{\mu}] \geq 0$. \qed

  %

\end{proof}

\subsection{Shift Model Proofs} \label{app:shiftproofs}

\begin{proof} [(Corollary \ref{cor:shift})]
  For general $n$, it suffices to show that $\prod_{i=1}^{n-1}
  (1-2^{-(n+1-i)}) \geq \frac{1}{2}$.
  \begin{align*}
    \prod_{i=1}^{n-1} (1-2^{-(n+1-i)}) &= \prod_{i=2}^n (1 - 2^{-i}) \\
    &= \prod_{i=2}^n \frac{1}{1 + \frac{2^{-i}}{1 - 2^{-i}}} \\
    &\geq \prod_{i=2}^n \frac{1}{\exp(\frac{2^{-i}}{1 - 2^{-i}})} \\
    &= \exp \left(- \sum_{i=2}^n \frac{2^{-i}}{1 - 2^{-i}} \right) \\
    &\geq \exp \left(- \sum_{i=2}^n \frac{2^{-i}}{1 - 2^{-2}} \right) \\
    &= \exp \left( - \frac{4}{3} \cdot \frac{1}{4} \sum_{i=0}^{n-2} 2^{-i} \right) \\
    &\geq \exp \left(- \frac{2}{3} \right) \\
    &> \frac{1}{2}.
  \end{align*}

  To check the value of $c(2)$, we simply plug in $n=2$.  That is,
  $\frac{2}{(1-2^{-(2+1-1)})} = \frac{8}{3}$.
\end{proof}

\subsection{Proofs of Final Theorems}
\label{app:finalthms}

\begin{proof} [(Theorem \ref{thm:identical})]
  Recall from Corollary \ref{cor:shift} that
  \begin{equation*}
    \Pr[A(\Gamma)] = c(n) \cdot 2^{-\binom{n+1}{2}} \cdot
    \sum_{\sigma \in S_n} \prod_{i=1}^{n-1} 2^{-(n-i) \gamma_{\sigma(i)}}.
  \end{equation*}
  Our goal is to average over the summation of permutations. Since we
  are treating $\Gamma$ as a random variable,
  \begin{align*}
    \Pr[A(\Gamma)] &= \E_{\Gamma} [c(n) \cdot 2^{-\binom{n+1}{2}} \cdot
      \sum_{\sigma \in S_n} \prod_{i=1}^{n-1} 2^{-(n-i) \gamma_{\sigma(i)}}] \\
    &= c(n) \cdot 2^{-\binom{n+1}{2}} \cdot
    \sum_{\sigma \in S_n} \E_{\Gamma} [\prod_{i=1}^{n-1} 2^{-(n-i) \gamma_{\sigma(i)}}]
  \end{align*}
  Then
  \begin{equation*}
    \E_{\Gamma} [\prod_{i=1}^{n-1} 2^{-(n-i) \Gamma_{\sigma(i)}}]
    = \sum_{\Gamma} \prod_{i=1}^{n-1} 2^{-(n-i) \Gamma_{\sigma(i)}} \cdot \Pr[B_{\Gamma}].
  \end{equation*}
  Let $\sigma(\Gamma) : \mathbb{N}^{n} \rightarrow \mathbb{N}^{n}$ be
  the operation mapping $\Gamma$ to $\Gamma'$ with entries permuted by
  $\sigma$. Define the inverse $\sigma^{-1}(\Gamma)$ accordingly. Then
  note
  \begin{multline*}
    \sum_{\Gamma} \prod_{i=1}^{n-1} 2^{-(n-i) \Gamma_{\sigma(i)}} \cdot \Pr[B_{\Gamma}] \\
    = \sum_{\Gamma'} \prod_{i=1}^{n-1} 2^{-(n-i) \Gamma'_{i}} \cdot \Pr[B_{\sigma^{-1}(\Gamma')}].
  \end{multline*}
  But because threads are distributed identically, the distribution of $\Gamma$
  is symmetric over any ordering $\sigma$: $\Pr[B_{\Gamma}] = \prod_{k=1}^{n} \Pr[B^{(k)}_{\Gamma_k}]
  = \prod_{k=1}^{n} \Pr[B^{(k)}_{\sigma(\Gamma)_k}] = \Pr[B_{\sigma(\Gamma)}]$.

  Hence we may write
  \begin{align*}
    \sum_{\Gamma'} \prod_{i=1}^{n-1} 2^{-(n-i) \Gamma'_{i}} &\cdot \Pr[B_{\sigma^{-1}(\Gamma')}] \\
    &= \sum_{\Gamma'} \prod_{i=1}^{n-1} 2^{-(n-i) \Gamma'_{i}} \cdot \Pr[B_{\Gamma'}] \\
    &= \E_{\Gamma'} [\prod_{i=1}^{n-1} 2^{-(n-i) \Gamma'_{i}}].
  \end{align*}
  There are $n!$ permutations, hence this proves the claim.
\end{proof}

\begin{proof} [(Theorem \ref{thm:2threads})]
  First observe that for any $\Gamma$ consisting of two segments, by
  Corollary \ref{cor:shift},
  \begin{align*}
    \Pr[A(\Gamma)] &= c(2) \cdot 2^{-\binom{2+1}{2}} \cdot \sum_{\sigma}
    \prod_{i=1}^{2-1} 2^{-(2-i) \gamma_{\sigma(i)}} \\
    &= \frac{8}{3} \cdot 2^{-3} \cdot \sum_{\sigma} 2^{- \gamma_{\sigma(i)}} \\
    &= \frac{1}{3} \cdot (2^{-\gamma_1} + 2^{-\gamma_2}).
  \end{align*}
  We then let $\Gamma$ be distributed as the critical windows of two reordered
  copies of an identical random program.
  \begin{equation*}
  \Pr[A] = \E_{\Gamma} [\Pr[A(\Gamma)]]
  = \E_{\Gamma} [\frac{1}{3} \cdot (2^{-\Gamma_1} + 2^{-\Gamma_2})]
  = \frac{2}{3} \cdot \E_{\Gamma} [2^{-\Gamma_1}].
  \end{equation*}
  Note that $B_{\gamma}$ is the event that $\gamma$ instructions end
  up between the critical $\ld$ and $\st$ exclusively, yet the
  critical window includes the critical $\ld$ and $\st$. Hence
  \begin{equation*}
  \E_{\Gamma} [2^{-\Gamma_1}] = \sum_{k=2}^{\infty} 2^{-k} \cdot \Pr[B_{k-2}].
  \end{equation*}

  \textbf{Sequential Consistency: } We first analyze the probability
  of bug manifestation in sequential consistency, the strictest of all
  memory models.  In sequential consistency, \emph{no} thread ever
  reorders. Hence no instructions ever appear between the critical
  $\ld$ and $\st$ in a given thread.  Thus
  \begin{equation*}
  \Pr[A] = \frac{2}{3} \cdot \E_{\Gamma} [2^{-\Gamma_1}]
  = \frac{2}{3} \cdot \frac{1}{4} = \frac{1}{6}.
  \end{equation*}

  \textbf{Weak Ordering: } Under the weakest memory model,
  instructions have a chance to bubble up regardless of whether its
  preceding instruction is a $\ld$ or $\st$. For this reason, the
  final size of the critical window is \emph{independent} of the
  original program, as it was in sequential consistency.

  Recall from Theorem \ref{thm:crit} that under Weak Ordering,
  \begin{equation*}
    \Pr[B_t] = \frac{2^{-t}}{3}
  \end{equation*}
  if $t > 0$, and $\Pr[B_t = 0] = \frac{2}{3}$. Thus
  \begin{align*}
    \E[2^{- \Gamma_1}] &= \sum_{t = 2}^{\infty} \Pr[B_{t-2}] \cdot 2^{-t} \\
    &= (2/3) \cdot 2^{-2} + \sum_{t = 3}^{\infty} \frac{2^{-(t-2)}}{3} \cdot 2^{-t} \\
    &= 2/12 + \frac{4}{3} \cdot \sum_{t = 3}^{\infty} 4^{-t} \\
    &= 1/6 + \frac{4}{3} \cdot \frac{1}{64} \cdot \frac{4}{3} \\
    &= \frac{7}{36}
  \end{align*}
  Then
  \[
  \Pr[A] = \frac{2}{3} \cdot \frac{7}{36} = \frac{7}{54}.
  \]
  Note that the probability of not manifesting has indeed decreased
  from sequential consistency. $\frac{1/6}{7/54} = 9/7.$
  Correct behavior is somewhat more likely than under sequential consistency.

  \textbf{Total Store Order: }
  We take advantage of the symmetry for $n=2$.
  We need not characterize the joint distribution of the lengths of two critical
  windows, because only the starting position of the lower window matters.

  Recall that Theorem \ref{thm:crit} shows that:
  \begin{equation*}
    \Pr[B_0] = \frac{2}{3},
  \end{equation*}
  and
  \begin{equation*}
    \Pr[B_{\gamma}] = \frac{6}{7} \cdot 4^{-\gamma} + R(\gamma) \cdot 2^{-\gamma},
  \end{equation*}
  for some positive $R(\gamma) < \frac{2}{21}$.

  Hence
  \begin{align*}
    \E[2^{- \Gamma_1}] &= \sum_{t = 2}^{\infty} \Pr[B_{t-2}] \cdot 2^{-t} \\
    &= (2/3) \cdot 2^{-2} \\
    & \quad + \sum_{t = 3}^{\infty} \left( \frac{6}{7} \cdot 4^{-(t-2)} + R(t-2) \cdot 2^{-(t-2)} \right) \cdot 2^{-t} \\
    &= (1/6) + \left( \frac{6}{7} \cdot 16 \cdot 8^{-3} \cdot \frac{8}{7} \sum_{t=0}^{\infty} 8^{-t} \right) \\
    & \quad + \left( 4 \sum_{t=3}^{\infty} R(t-2) 4^{-t} \right) \\
    &= (1/6) + \frac{3}{98} + 4 \sum_{t=3}^{\infty} R(t-2) 4^{-t}.
  \end{align*}
  Plugging in $R(t) > 0$ gives
  \[
  \Pr[A] = \frac{2}{3} \cdot \E[2^{- \Gamma_1}] > \frac{2}{3} (\frac{1}{6} + \frac{3}{98}) = 58/441 > 0.1315.
  \]
  Similarly, plugging in $R(t) < 2/21$ gives
  \begin{align*}
    \Pr[A] &= \frac{2}{3} \cdot \E[2^{- \Gamma_1}] \\
    &< 58/441 + \frac{2}{3} \cdot (4 \cdot \frac{2}{21} \cdot 4^{-3} \cdot \frac{4}{3}) \\
    &= 58/441 + 1/189 \\
    &< 0.1369.
  \end{align*}

  We now see that with two threads, the probability of reliable
  execution is substantially closer to that of weak ordering (0.1296)
  than that of sequential consistency (0.1666).
\end{proof}

\begin{proof} [(Theorem \ref{thm:nthreads})]
  We first analyze the probability for Sequential Consistency. As the
  strongest memory model, the probability of successful execution
  serves as an upper bound for every other model. This is because the
  likelihood that the shift process results in an overlap is
  monotonically increasing in the distribution of critical window
  size.

  \textbf{Sequential Consistency: }
  Again, recall from Corollary \ref{cor:shift} that
  \begin{align*}
    \Pr[A] &= \E_{\Gamma} [\Pr[A(\Gamma)]] \\
      &= c(n) \cdot 2^{-\binom{n+1}{2}} \cdot \sum_{\sigma}
      \E_{\Gamma} \left[\prod_{i=1}^{n-1} 2^{-(n-i) \gamma_{\sigma(i)}}\right].
  \end{align*}
  Under sequential consistency, $\gamma_{\sigma(i)}=2$ always. Hence
  \begin{align*}
    \Pr[A] &= c(n) \cdot 2^{-\binom{n+1}{2}}
    \cdot \sum_{\sigma} \prod_{i=1}^{n-1} 2^{-(n-i) 2} \\
    &= c(n) \cdot 2^{-\binom{n+1}{2}} \cdot n! \cdot 2^{-2\binom{n}{2}} \\
    &= 2^{-n^2 (3/2 + o(1))},
  \end{align*}
  where the last line follows from Stirling's formula:
  \begin{align*}
    n! &= \sqrt{2 \pi n} (\frac{n}{e})^{n} (1 + o(1)) \\
    &= \exp \left( \frac{\ln (2 \pi)}{2} + \frac{\ln n}{2} + (\ln n - 1) n \right) (1+o(1)) \\
    &= e^{(n \ln n) (1+ o(1))} \\
    &= e^{n^2 \cdot o(1)}.
  \end{align*}

  \textbf{Other Models: } Surprisingly, to achieve the same bound for
  any model, all we need is a lower bound on the probability of
  generating a small critical window.

  \begin{claim}
    In every memory model,
    \[
    \Pr[B_{0}] \geq \frac{1}{2}.
    \]
  \end{claim}
  \begin{proof}
    This can be observed by the fact that no matter the model, the
    critical $\ld$ must move up to have any chance of the critical
    window growing. But even if the critical $\ld$ is allowed to pass
    the instruction above it, this only occurs with probability 1/2.
  \end{proof}
  The claim has the following consequences.
  \[
  \Pr[\bigwedge_{i=1}^{n-1} \gamma_{i} = 2] \geq 2^{-(n-1)},
  \]
  hence
  \[
  \Pr[\prod_{i=1}^{n-1} 2^{(n-i) \gamma_{\sigma(i)}} = 2^{-2\binom{n}{2}}] \geq 2^{-(n-1)},
  \]
  thus
  \[
  \E[\prod_{i=1}^{n-1} 2^{(n-i) \gamma_{\sigma(i)}}] \geq 2^{-2\binom{n}{2} - (n-1)}.
  \]
  Plugging this expectation into the probability of correct execution
  again gives
  \[
  \Pr[A] > c(n) \cdot 2^{-\binom{n+1}{2}} \cdot n! \cdot 2^{-2\binom{n}{2} - (n-1)}
  = 2^{-n^2 (3/2)}.
  \]
  Recall that Sequential Consistency offers the largest probability of
  correct execution of any model. Hence upper bounding the above value
  by the probability for Sequential Consistency, we have
  \[
  \Pr[A] < 2^{-n^2 (3/2 + o(1))}.
  \]

  This completes the proof. \qed
\end{proof}
}
{}

\end{document}